\documentclass[11pt]{article}

\usepackage[margin=1in]{geometry}
\usepackage[hyphens]{url}
\usepackage[colorlinks=true, allcolors=blue, breaklinks=true]{hyperref}
\usepackage{amsfonts,mathrsfs,mathpazo,eucal,xspace,graphicx,appendix,xcolor,array}
\usepackage{endnotes}
\usepackage{amssymb,latexsym}
\usepackage{enumitem}
\setlist{itemsep=0mm}
\usepackage{float}
\usepackage{lscape}

\usepackage{amsmath, mathtools, physics}
\usepackage{complexity}

\newcommand\numberthis{\addtocounter{equation}{1}\tag{\theequation}}

\usepackage[linesnumbered,ruled,vlined]{algorithm2e}

\usepackage{amsthm}
\newtheorem{theorem}{Theorem}
\newtheorem*{theorem*}{Theorem}

\newtheorem{lemma}[theorem]{Lemma}
\newtheorem*{lemma*}{Lemma}

\theoremstyle{definition}
\newtheorem{definition}[theorem]{Definition}

\newtheorem*{remark*}{Remark}

\newcommand{\NN}{\ensuremath{\mathbb{N}}}

\newcommand{\RR}{\ensuremath{\mathbb{R}}}

\newcommand{\projerror}{g(n)}
\newcommand{\depthname}{depth-$d$ }

\newcommand{\nnote}[1]{}
\newcommand{\mnotetwo}[1]{}
\definecolor{darkbrown}{RGB}{101,67,33}
\allowdisplaybreaks

\title{Approximating Output Probabilities of Shallow Quantum Circuits which are Geometrically-local in any Fixed Dimension.}
\author{Suchetan Dontha, \, Shi Jie Samuel Tan, \, Stephen Smith, \\ Sangheon Choi, \, Matthew Coudron}
\date{}

\begin{document}

\maketitle

\begin{abstract}

    We present a classical algorithm that, for any $D$-dimensional geometrically-local,  quantum circuit $C$ of polylogarithmic-depth, and any bit string $x \in \{0,1\}^n$, can compute the quantity $|\bra{ x}C   \ket{0^{\otimes n}}|^2$ to within any inverse-polynomial additive error in quasi-polynomial time, for any fixed dimension $D$.  This is an extension of the result \cite{CC20}, which originally proved this result for $D = 3$.  To see why this is interesting, note that, while the $D=1$ case of this result follows from a standard use of Matrix Product States, known for decades, the $D=2$ case required novel and interesting techniques introduced in \cite{BGM19}.  Extending to the case $D=3$ was even more laborious, and required further new techniques introduced in  \cite{CC20}.  Our work here shows that, while handling each new dimension has historically required a new insight, and fixed algorithmic primitive, based on known techniques for $D \leq 3$, we can now handle any fixed dimension $D > 3$.  
    
    Our algorithm uses the Divide-and-Conquer framework of \cite{CC20} to approximate the desired quantity via several instantiations of the same problem type, each involving $D$-dimensional circuits on about half the number of qubits as the original.  This division step is then applied recursively, until the width of the recursively decomposed circuits in the $D^{th}$ dimension is so small that they can effectively be regarded as $(D-1)$-dimensional problems by absorbing the small width in the $D^{th}$ dimension into the qudit structure at the cost of a moderate increase in runtime.  The main technical challenge lies in ensuring that the more involved portions of the recursive circuit decomposition and error analysis from \cite{CC20} still hold in higher dimensions, which requires small modifications to the analysis in some places.  Our work also includes some simplifications, corrections and clarifications of the use of block-encodings within the original classical algorithm in \cite{CC20}.

\end{abstract}
\section{Introduction}
    
    It is known that it is $\#P$-hard to compute the quantity  $|\bra{ x}C   \ket{0^{\otimes n}}|^2$  to within $2^{-n^2}$ additive error for low-depth, geometrically-local quantum circuits $C$ \cite{ Mov19, KMM21}, and worst-case hardness results for this task date back to \cite{TD04}.  These hardness results indicate that computing output probabilities with such small additive error is almost certainly out of reach for both classical \emph{and} quantum computers.  If we restrict our attention to additive errors that are achievable with quantum computers, such as inverse polynomial additive error achievable by taking polynomially many samples from the quantum circuit $C$, then classical hardness for this estimation problem is much less clear.  In fact, \cite{BGM19} introduced an elegant classical polynomial time algorithm for this estimation task in the case of 2D circuits.  Their algorithm makes a novel use of 1D Matrix Product States carefully tailored to the 2D geometry of the circuit in question.  While it is not clear how to generalize the techniques of \cite{BGM19} to higher dimensional circuits, \cite{CC20} introduced a Divide-and-Conquer algorithm that can compute the quantity $|\bra{ x}C   \ket{0^{\otimes n}}|^2$ to within any inverse-polynomial additive error in quasi-polynomial time for any $3D$, constant-depth quantum circuit $C$.  The algorithm in \cite{CC20} works by recursively subdividing the quantum circuit $C$ into pieces, constructed using block-encodings, and introduces new techniques for analyzing the extent to which quantum entanglement between different qubits can impact the global quantity $|\bra{ x}C   \ket{0^{\otimes n}}|^2$.   
    
    Given the progression of ideas required to classically approximate the output probabilities of higher dimensional quantum circuits, it is natural to wonder what would be required to go even further.  In this work we will show that there exists a classical quasi-polynomial time algorithm which can compute $|\bra{ x}C   \ket{0^{\otimes n}}|^2$ to inverse polynomial additive error for any constant-depth, geometrically-local quantum circuit of fixed dimension $D$.

    \begin{theorem}[Main Result]
    For any $D$-dimensional geometrically-local, depth $d$ quantum circuit $C$ acting on $n$ qubits, the algorithm $\mathcal{A}_{full}(S=(C,L,M,N), \mathcal{B}, \delta, D)$ computes the quantity $|\bra{ x}C   \ket{0^{\otimes n}}|^2$ to within $\delta$ additive error in time
    $\delta^{-2} \cdot2^{O((d\polylog(n))^{D \cdot 3^{D}}) (1/\delta)^{1 / \log^2(n)}}$. \footnote{For clarity we assume that the $n$ qubits are arranged in a perfect D-dimensional cubic lattice.  Here $S=(C,L,M,N)$ is the synthesis describing circuit $C$, as defined in this paper and in \cite{CC20}, and $\mathcal{B}$ is our base-case algorithm which we specify to be the 2D algorithm of \cite{BGM19}, and which our algorithm uses to solve subproblems which have been recursively subdivided down to 2 dimensions. }
    \end{theorem}

    A key motivation for generalizing simulation results to higher dimensions exists at the level of techniques.  Historically, the simulation of low-depth and geometrically local quantum circuits has required a new mathematical innovation every time the dimension, $D$, of the geometric locality is increased.  The $D=1$ case is solved using the famous technique of Matrix Product States (MPS), which is fundamental to the field and has been known for decades.  However, it was not until recently that an algorithm was discovered for estimating output amplitudes in the case $D=2$, and it requires a novel technique beyond standard MPS \cite{BGM19}.  Finding an algorithm for the $D=3$ case, \cite{CC20}, required a completely different approach, this time departing from the paradigm of MPS, and requiring 50 pages of mathematics to formalize a divide-and-conquer algorithm.  Our result shows that this trend of requiring completely new techniques to extend from $D$ to $D+1$ need not continue.  One fixed divide-and-conquer algorithmic primitive, allows us to inductively establish an additive-error classical simulation algorithm for any dimension $D$. 
    
    Note that, while our algorithm runs in quasi-polynomial time in $n$ for any fixed $D$, the runtime is triply exponential in the dimension $D$. If we set $D = O(\log(\log(\polylog(n))))$ and $\delta$ to be inverse quasi-polynomial, then the algorithm still runs in quasi-polynomial time on a constant depth geometrically local circuit. In particular, this means that the algorithm can approximate the output probabilities of any constant depth quantum circuit that is geometrically local in $O(\log(\log(\polylog(n))))$ dimensions. It is, therefore, interesting to consider the computational complexity of this problem as a function of $D$, since this could shed light on the extent to which arbitrary low-depth quantum circuits can be efficiently simulated.  As an extreme example, an algorithm which had runtime polynomial in $D$ could be used efficiently on constant depth quantum circuits which are not geometrically local at all.  This is because any constant depth quantum circuit on $n$-qubits can be considered to be geometrically local in dimension $D = n$.  We do not expect that our current approach can achieve a runtime polynomial in $D$, but we believe that even a runtime that is singly exponential in $D$, allowing the simulation of circuits which are geometrically local in dimension $D = \log(n)$, could have practically relevant consequences.  We leave, as an open problem the question of the optimal $D$-dependence for algorithms simulating constant-depth geometrically-local quantum circuits.
    
    Our paper is organized as follows: In section 2 we review block-encodings and syntheses, both of which are used extensively throughout the algorithm. Note that section 2 primarily consists of definitions and lemmas from \cite{CC20} that are tweaked for clarity and correctness. In section 3 we provide the pseudocode for our algorithm and prove our main result. The runtime and error analysis for our algorithm are located in sections 3.1 and 3.2 respectively.  

\section{Block-encodings and Syntheses}
\label{sec:preliminaries}

In order to state the pseudocode for our algorithms in Section \ref{sec:algorithms} below we first need to establish a way to construct the ``recursive subdivisions" of the quantum circuit $C$ that our divide-and-conquer algorithm iteratively creates.  We will concretely describe these subdivisions as ``syntheses", as defined in \cite{CC20} and reviewed here for the convenience of the reader.  Syntheses themselves use the idea of a block-encoding which we paraphrase below from \cite{GSLW18}.

\textbf{In order to understand the following discussion, which is essential to the rest of this paper, it is necessary to read sections 2, 3, and 4 of \cite{CC20}. The lemmas repeated in this section are only included here in order to clarify or correct certain definitions in section 3 of \cite{CC20}.  Many other definitions and lemmas from  sections 2, 3, and 4 of \cite{CC20} are not repeated here and must be read from the original document (see the arxiv verison at \url{https://arxiv.org/pdf/2012.05460.pdf}). }

\begin{definition}[Block-encoding]
     Suppose that $A$ is an $s$-qubit operator, $\alpha, \epsilon \in \RR_+$ and $a\in\NN$. Then we say that the $(s+a)$-qubit unitary $U$ is an $(\alpha, a, \epsilon)$-block-encoding of $A$, if 
     \[
     \norm{A-\alpha(\bra{0}^{\otimes a}\otimes I)U(\ket{0}^{\otimes a}\otimes I)}\le \epsilon.
     \]
\end{definition}

Consider a cut $B\cup M\cup F$ made anywhere in the cube and let $\sigma_{M\cup F}=$\\$\tr_B\left(C_{B\cup M\cup F}\ketbra{0}{0}_{B\cup M\cup F}C_{B\cup M\cup F}^\dagger\right)$.
The following result is obtained by applying Lemma 45 of \cite{GSLW18}:

\begin{lemma}[Block-encoding for $\sigma_{M\cup F}$]\label{def:sigma}
    The following is a $(1, |B\cup M\cup F|, 0)$-block-encoding of $\sigma_{M\cup F}$:
    \[\Gamma=
    \left(C_{B\cup M'\cup F'}^\dagger \otimes I_{M \cup F}\right)\left(I_{B} \otimes SWAP_{M \cup F,M' \cup F'}\right)\left(C_{B\cup M'\cup F'}\otimes I_{M \cup F}\right).
    \]
    In the above, $C_{B\cup M'\cup F'}$ is notation to indicate that we will be applying the circuit $C_{B\cup M\cup F}$ on the registers $B$, $M'$, and $F'$. In other words,
    \[
    \sigma_{M\cup F}=\left(\bra{0}_{B\cup M'\cup F'}\otimes I_{M\cup F}\right) \Gamma \left(\ket{0}_{B\cup M'\cup F'}\otimes I_{M\cup F}\right).
    \]
\end{lemma}

The registers $M'$ and $F'$ above are copies of the registers $M$ and $F$, respectively, and are introduced by Lemma 45 of \cite{GSLW18}. By interleaving $M'$ with $M$ and $B'$ with $B$ and adding swap gates where appropriate, we can ensure that the resulting circuit, $\Gamma$, is still geometrically-local and has depth at most $3$ times the depth of $C_{B\cup M\cup F}$. By simply moving the $M$ register in Lemma \ref{def:sigma} to the set of registers which are post-selected, we see that $\Gamma$ is also a block-encoding of $\rho_F := \bra{0}_M \sigma_{M\cup F}\ket{0}_M$.
\newpage
\begin{lemma}[Block-encoding for $\rho_F$]
    The block-encoding introduced in Lemma \ref{def:sigma}, $\Gamma$, is a\\
    $(1, |B\cup F| + 2|M|, 0)$-block-encoding of $\rho_F$. Note that Lemma 4 is a correction of Lemma 7 of \cite{CC20}.
\end{lemma}
\begin{proof}
\begin{align*}
    (&\bra{0}_{B\cup M' \cup F' \cup M}\otimes I_F)\Gamma (\ket{0}_{B\cup M' \cup F' \cup M}\otimes I_F)\\
    &=\bra{0}_M(\bra{0}_{B\cup M' \cup F'}\otimes I_F) \Gamma (\ket{0}_{B\cup M' \cup F'}\otimes I_F)\ket{0}_M\\
    &=\bra{0}_M\sigma_{M\cup F}\ket{0}_M\\
    &=\rho_F.
\end{align*}
\end{proof}

Since $\rho_F$ is the state that we are really interested in, we will henceforth refer to $\Gamma$ as $\Gamma_{\rho_F}$. We can now iteratively apply Lemma 53 from \cite{GSLW18} to obtain a block-encoding for $\rho_F^k$ for any integer $k\ge 1$. To do this, we will need $k-1$ copies of each of the registers $B$, $M'$, $F'$, and $M$. Let $B_1=B$, $M'_1=M'$, $F'_1 = F'$, and $M_1 = M$. Furthermore, for each $i$, $1<i\le k$, let $B_i$, $M'_i$, $F'_i$, $M_i$ be copies of $B$, $M'$, $F'$, and $M$, respectively.

\begin{lemma}[Block-encoding for $\rho_F^k$]\label{lem:powerofrhoF}
     The following is a $(1, k|B\cup M'\cup F'\cup M|, 0)$-block-encoding of $\rho_F^k$:
     \[\Gamma_{\rho_F^k}=
    \prod_{i=1}^k\left(C_{B_i\cup M'_i\cup F'_i}^\dagger \otimes I_{M_i \cup F}\right)\left(I_{B_i} \otimes SWAP_{M_i \cup F,M_i' \cup F_i'}\right)\left(C_{B_i\cup M_i'\cup F_i'}\otimes I_{M_i \cup F}\right).
    \]
    In other words,
    \[
    \rho_F^k = \left(\bra{0}_{\mathcal{B}_k \cup \mathcal{M}_k' \cup \mathcal{F}_k' \cup \mathcal{M}_k} \otimes I_{F}\right)\Gamma_{\rho_F^k} \left(\ket{0}_{\mathcal{B}_k \cup \mathcal{M}_k' \cup \mathcal{F}_k' \cup \mathcal{M}_k} \otimes I_{F}\right)
    \]
    where $\mathcal{B}_k = B_1\cup B_2\cup \cdots \cup B_k$, $\mathcal{M}_k' = M_1'\cup M_2'\cup \cdots \cup M_k'$, etc. Note that this is a correction of equation 7 of \cite{CC20}
\end{lemma}

\begin{lemma}[Block-encoding for $\rho_B^k$]\label{lem:powerofrhoB}
     Analogously, the following is a $(1, k|F\cup M'\cup B'\cup M|, 0)$-block-encoding of $\rho_B^k$:
     \[\Gamma_{\rho_B^k}=
    \prod_{i=1}^k\left(C_{B'_i\cup M'_i\cup F_i}^\dagger \otimes I_{M_i \cup B}\right)\left(I_{F_i} \otimes SWAP_{M_i \cup B,M_i' \cup 
    B_i'}\right)\left(C_{B'_i\cup M'_i\cup F_i}\otimes I_{M_i \cup B}\right).
    \]
    In other words,
    \[
    \rho_B^k = \left(\bra{0}_{\mathcal{F}_k \cup \mathcal{M}_k' \cup \mathcal{B}_k' \cup \mathcal{M}_k} \otimes I_{B}\right)\Gamma_{\rho_B^k} \left(\ket{0}_{\mathcal{F}_k \cup \mathcal{M}_k' \cup \mathcal{B}_k' \cup \mathcal{M}_k} \otimes I_{B}\right)
    \]
    where $\mathcal{F}_k = F_1\cup F_2\cup \cdots \cup F_k$, $\mathcal{M}_k' = M_1'\cup M_2'\cup \cdots \cup M_k'$, etc. \\
    Note that this is a correction for equation 7 of \cite{CC20}
\end{lemma}

Importantly, we are free to interleave all of the copies of the registers $B$, $M$ and $F$ with their originals. We do this in such a way so that we can minimally pad each 2-qubit gate from $C_{B\cup M\cup F}$ with swap gates so that this new 'padded' circuit is still geometrically local. Furthermore, the depth of this new padded circuit is at most $(2k+1)$ times the original depth of $C$.

\begin{definition}[Synthesis] \label{def:synth}
 	We say that an unnormalized quantum state $\phi$ is \emph{synthesized} by a quantum circuit $\Gamma$, if $\Gamma$ has three registers of qubits $L, M, N$ such that:
 	
 	\begin{align}
 	\phi = \phi_{(\Gamma, L, M, N)}= \tr_{L\cup M}(  \bra{0_M}\Gamma\ket{0_{L \cup M \cup N}}\bra{0_{L \cup M \cup N}}\Gamma^{\dagger}\ket{0_M}).
 	\end{align}

 	In this case we say that the circuit $\Gamma$ together with a specification of the registers $L,M,N$ constitutes a \emph{synthesis} of $\phi$.  When $\phi$ is implicit we will call this collection $(\Gamma, L, M, N)$ a \emph{synthesis}. This definition was taken directly from \cite{CC20} and is only here for the convenience of the reader. All syntheses explicitly used in the rest of this paper are defined in section 4 of \cite{CC20}.
 \end{definition}

\section{Algorithms and Analysis} \label{sec:algorithms}

Having discussed the essential concepts of syntheses and block-encodings in Section \ref{sec:preliminaries} above, we now give an explicit description of our classical simulation algorithm below.  Our algorithm is divided into two pieces, Algorithm \ref{alg:constant-width-assumption-driver2} and Algorithm \ref{alg:constant-width-assumption2}.  Algorithm \ref{alg:constant-width-assumption-driver2} simply handles some technical edge cases for the error parameter $\delta$, and sets the stage for making a call to Algorithm \ref{alg:constant-width-assumption2}.  Algorithm \ref{alg:constant-width-assumption2} contains the actual divide-and-conquer structure, describing how to perform recursive calls to itself and Algorithm \ref{alg:constant-width-assumption-driver2} in one dimension lower.

The following theorem and lemmas state and prove our main result by giving runtime bounds and error bounds for Algorithm \ref{alg:constant-width-assumption-driver2}. Algorithm \ref{alg:constant-width-assumption-driver2} is defined in complete pseudo-code below, for any dimension $D$, and our main result is proved by induction on dimension $D$.

 \begin{theorem} \label{thm:inductivestep}
    For any $D$-dimensional geometrically-local, depth $d$ quantum circuit $C$ acting on $n$ qubits, the algorithm $\mathcal{A}_{full}(S=(C,L,M,N), \mathcal{B}, \delta, D)$ computes the quantity  $|\bra{0^{\otimes n}}C\ket{0^{\otimes n}}|^2$ to within $\delta = 1 / n^{\log(n)}$ additive error in time
    $2^{O((d\polylog(n))^{D \cdot 3^{D}})}$. Furthermore, let $w_{D + 1}, w_{D + 2}, \ldots$ be the widths of the qubit array in dimensions $D + 1, D + 2, \ldots$ respectively. Then for any geometrically-local, depth d quantum circuit $C$ acting on a lattice of $n$ qubits having side length at most $w_{D+i}$ in dimension $D+i$, the algorithm $\mathcal{A}_{full}(S=(C,L,M,N), \mathcal{B}, \delta, D)$ computes the quantity $|\bra{0^{\otimes n}}C\ket{0^{\otimes n}}|^2$ to within $\delta = 1 / n^{\log(n)}$ additive error in time $2^{O((d\polylog(n))^{D3^{D}} w^{1/3})}$, where $w \equiv \prod_{i=1}^{\infty}w_{D+i}$.\footnote{We assume the $n$ qubits are arranged in such a way that the length of each edge of the qubit lattice is $O(n^{1/D})$}
    \end{theorem}
\begin{proof}
We will prove Theorem \ref{thm:inductivestep} by induction on the dimension $D$.  For the base-case, $D=3$, this theorem is a direct consequence of the main result of \cite{CC20}.  For $D > 3$, assuming, by induction, that we have already established Theorem \ref{thm:inductivestep} for dimension $D-1$, the dimension $D$ version of the Theorem follows by Lemmas \ref{thm:mainruntimelemma} and \ref{thm:mainerrorlemma} respectively.  The key inductive step in those two analyses happens at the point in the analysis where Algorithm \ref{alg:constant-width-assumption2} makes calls, such as $\mathcal{A}_{full}(S_{i,j}, \mathcal{B},\epsilon, D-1)$, to a $D-1$ dimensional version of Algorithm \ref{alg:constant-width-assumption-driver2}.  At those points the runtime and error guarantees for the $D-1$ dimensional version of $\mathcal{A}_{full}$ that are required by the analyses in Lemmas \ref{thm:mainruntimelemma} and \ref{thm:mainerrorlemma} are ensured by the inductive assumption that Theorem \ref{thm:inductivestep} already holds for the $D-1$ dimensional case.
\end{proof}

\begin{lemma} \label{thm:mainruntimelemma}
Let $w$ be defined as in Theorem \ref{thm:inductivestep}. Then $\mathcal{A}_{full}(S=(C,L,M,N), \mathcal{B}, \delta, D)$ runs in time
    $\delta^{-2} \cdot2^{O((d\polylog(n))^{D \cdot 3^{D}}w^{1/3}) (1/\delta)^{1 / \log^2(n)}}$.
\end{lemma}
\begin{proof}
Refer to Appendix \ref{appendix:lemma-proofs} for proof.
\end{proof}

\begin{lemma} \label{thm:mainerrorlemma}
    $\mathcal{A}_{full}(S=(C,L,M,N), \mathcal{B}, \delta, D)$ returns an $\delta$-additive error approximation of $|\bra{0_{ALL}}C\ket{0_{ALL}}|^2$ 
\end{lemma}
\begin{proof}
Refer to Appendix \ref{appendix:lemma-proofs} for the proof.
\end{proof}

\newpage

\begin{algorithm}[H] \label{alg:constant-width-assumption-driver2}  \label{alg:quasi-poly-driver}
	\SetKwInOut{Input}{Input}\SetKwInOut{Output}{Output}
	\Input{Synthesis S = (C, L, M, N) where $C$ is a D-Dimensional Geometrically-Local \depthname circuit, $\mathcal{B}$ a base case algorithm for 2D circuits, approximation error $\delta$, dimension $D$}
	\Output{An approximation of  $|\bra{0_{ALL}}C\ket{0_{ALL}}|^2$ to within additive error $\delta$.} 
	
	\tcc{We begin by handling the case in which $\delta$ is so small that it trivializes our runtime, and the case in which $\delta$ is so large that it causes meaningless errors:}
	
	\If{$\delta \leq 1/n^{\log^2(n)}$  \label{ln:checkforsmalldelta}}{\Return The value $|\bra{0_{ALL}}C\ket{0_{ALL}}|^2$ computed with zero error by a ``brute force" $2^{O(n)}$-time algorithm. }
		
		\textbf{if}  $\delta \geq 1/2$ \label{ln:deltalarge} \textbf{then}  \textbf{return} $1/2$
	
	\If{D = 2}{\Return $\mathcal{B}(S, \delta)$}
	\tcc{Here begins the non-trivial part of the algorithm:}

	Let $N$ be the register containing all of the qubits on which $C$ acts.  Since these qubits are arranged in a hyper-cubic lattice, the sides of the hyper-cube $N$ must have length $n^{\frac{1}{D}}$.  We will call the length of this side the ``width" and will now describe how to ``cut" the hyper-cube $N$, and the circuit $C$, perpendicular to this particular side.    
	
	Select $\frac{1}{10d}n^{\frac{1}{D}}$  light-cone separated slices $K_i$ of $10d$ width in $N$, with at most 10$d$ distance between adjacent slices. \label{ln:startslicesearch}  Let $h(n) = \log^7(n)$.  Run Algorithm  $\mathcal{A}_{full}(S,\mathcal{B}, 2^{\frac{\log(\delta)}{2h(n)} - 1} - 2^{\frac{\log(\delta)}{h(n)} - 1}, D-1)$ to check if at least $\frac{1}{10d}n^{\frac{1}{D}} - h(n)$ of the slices obey:  \label{ln:endslicesearch} 
	
	\begin{equation} 
	\abs{\tr \left (\bra{0_{M_i}}C\ket{0_{ALL}}\bra{0_{ALL}}C^{\dagger}\ket{0_{M_i}}\right)}  \geq 2^{\frac{\log(\delta)}{h(n)}}. \nonumber 
	\end{equation}

	OR, there are fewer than $\frac{1}{10d}n^{\frac{1}{D}} - h(n)$ slices that obey: \label{ln:doubleendslicesearch}

\begin{equation} 
\abs{\tr \left (\bra{0_{M_i}}C\ket{0_{ALL}}\bra{0_{ALL}}C^{\dagger}\ket{0_{M_i}}\right)}  \geq  2^{\frac{\log(\delta)}{h(n)}}. \nonumber \label{eq:heavyweight8}
\end{equation}

\tcc{  See the runtime analysis in the proof of Theorem 28 of \cite{CC20} for a detailed explanation of how the aforementioned run of $\mathcal{A}_{full}$ can efficiently distinguish between the above two cases (via Remark 6 in \cite{CC20}).}
	
		\textbf{if}  Fewer than $\frac{1}{10d}n^{\frac{1}{D}} - h(n)$ of the slices obey Line \ref{eq:heavyweight8} \label{ln:easyifalg1} \textbf{then return} 0  
	
	\If{At least $\frac{1}{10d}n^{\frac{1}{D}} - h(n)$ of the slices obey Line \ref{eq:heavyweight8} \label{ln:hardifalg1}}{

		We will denote the set of these slices by $K_{heavy}$.  Note that the maximum amount of width between any two adjacent slices in $K_{heavy}$ is $10d \cdot h(n)$.  Furthermore, the maximum amount of width collectively between $\Delta$ slices in $K_{heavy}$ is $10d\Delta + 10d \cdot h(n)$.  Now that the set $K_{heavy}$ has been defined, we will use this fixed set in the recursive algorithm, Algorithm \ref{alg:constant-width-assumption2}.
		
		\Return $\mathcal{A}(S, \eta = \frac{\log(n)}{D\log(4/3)} , \Delta = \log(n), \epsilon = \delta2^{-10 \log(n) \log(\log(n)))} , h(n) = \log^7(n), K_{heavy} , D,  \mathcal{B})$
	}

	\caption{$\mathcal{A}_{full}(S=(C,L,M,N), \mathcal{B}, \delta, D)$:  Quasi-Polynomial Time Additive Error Approximation for $|\bra{0_{ALL}}C\ket{0_{ALL}}|^2$.  }
\end{algorithm}

\begin{algorithm}\label{alg:constant-width-assumption2}\label{alg:quasi-poly-subroutine}
	\SetKwInOut{Input}{Input}\SetKwInOut{Output}{Output}
	\Input{D-dimensional Geometrically-Local, \depthname synthesis $S$, number of iterations $\eta$, number of cuts $\Delta$, positive base-case error bound $\epsilon > 0$, a set of heavy slices $K_{heavy}$, dimension $D$, $\mathcal{B}$ a base case algorithm for 2D circuits}
	\Output{An approximation of the quantity $\bra{0_{N}}\phi_S\ket{0_{N}}$ where $\phi_S$ is the un-normalized mixed state specified by the D-dimensional geometrically-local, \depthname synthesis $S$, and $\ket{0_{N}}$ is the $0$ state on the entire $N$ register of that synthesis.}

	Given the geometrically-local, \depthname synthesis $S = (\Gamma, L, M, N)$, let us ignore the registers $L$ and $M$ as they have already been measured or traced-out. 
	
	Let $\ell$ be the width of the $N$ register of the synthesis $S$.  Define the stopping width $w_0 \equiv 20d(\Delta+h(n) + 2)$.

	\If{ $\ell < w_0 = 20d(\Delta+h(n)+2)$  OR $\eta <1$ \label{ln:alg2stoppingwidth} }{Compute the quantity $\bra{0_{N}}\phi_S\ket{0_{N}}$ to within error $\epsilon$.\\
		\Return $\mathcal{A}_{full}(S, \mathcal{B}, \epsilon, D-1)$
	}
	
	\Else{ 	
		We will ``slice" the $D$-Dimensional geometrically-local, \depthname synthesis $S$ in $\Delta$ different locations, as follows:
		
		Since $N$ is $D$-Dimensional we define a region $Z \subset N$ to be the sub-hyper-cube of $N$ which has width $10d(\Delta+h(n)+2)$, and is centered at the halfway point of $N$ width-wise (about the point $\ell /2$ of the way across $N$). Since the maximum amount of width collectively between $\Delta$ slices in $K_{heavy}$ is $10d\Delta + 10d \cdot h(n)$ (see Algorithm \ref{alg:constant-width-assumption-driver2}), we are guaranteed that the region $Z$ will contain at least $\Delta$ slices, $K_1,K_2,\dots,K_\Delta$, from $K_{heavy}$.  For any two slices $K_i, K_j\in K_{heavy}$, let the un-normalized states $\ket{\varphi_{\L, i}}, \ket{\varphi_{i, j}}, \ket{\varphi_{j, \R}}$, and corresponding sub-syntheses $S_{\L, i}, S_{i,j}, S_{j, \R}$ be as defined in Definition 23 from \cite{CC20}, with $K = \log^3(n)$.   We will use these to describe the result of our division step below.\label{ln:alg2defz}
		
		For each $K_i \in K_{heavy}$ pre-compute the quantity $\kappa^{i}_{T, \epsilon}$, with $T = \log^3(n)$, and  $\epsilon = \delta 2^{-10 \log(n) \log(\log(n)))}$. \label{ln:compnorm}  \label{ln:pre-compute}
		
		\Return \label{ln:returnofA}\begin{align}
		&\sum_{i=1}^{\Delta}\frac{1}{(\kappa^{i}_{T, \epsilon})^{4K+1}} \mathcal{A}(S_{L,i},\eta-1) \cdot \mathcal{A}(S_{i,R},\eta-1) \label{eq:onecutterms}\\
		&-\sum_{i=1}^{\Delta}\sum_{j=i+1}^{\Delta}\frac{1}{(\kappa^{i}_{T, \epsilon}\kappa^{j}_{T, \epsilon})^{4K+1}}\mathcal{A}(S_{L,i},\eta-1) \cdot \mathcal{A}_{full}(S_{i,j}, \mathcal{B},\epsilon, D-1)\cdot \mathcal{A}(S_{j,R},\eta-1) \label{eq:twocutterms} \\
		&+\sum_{i=1}^{\Delta}\sum_{j=i+2}^{\Delta}\frac{1}{(\kappa^{i}_{T, \epsilon}\kappa^{j}_{T, \epsilon})^{4K+1}}\mathcal{A}(S_{L,i},\eta-1) \cdot \mathcal{A}(S_{j,R},\eta-1) \nonumber \\
		&\cdot\Bigg[\sum_{\sigma\in\mathcal{P}(\{i+1, \cdots, j-1\})\setminus\emptyset}(-1)^{\abs{\sigma}+1} \mathcal{A}_{full}\left ( \left (\otimes_{k \in \sigma} \Pi^K_{F_k} \bra{0_{M_k}} \right )\phi_{i,j}\left (\otimes_{k \in \sigma} \ket{0_{M_k}}\Pi^K_{F_k}\right ),\mathcal{B}, \frac{\epsilon}{2^\Delta}, D-1 \right) \Bigg]  \label{eq:basecaseinsum}
		\end{align}
		\tcc{Note that for brevity it is implied that $\mathcal{A}(S, \eta)=\mathcal{A}(S, \eta, \Delta, \epsilon, h(n),  K_{heavy}, D,  \mathcal{B})$.}

	}
	\caption{$\mathcal{A}(S, \eta, \Delta, \epsilon, h(n),  K_{heavy}, D, \mathcal{B})$: Recursive Divide-and-Conquer Subroutine for Algorithm \ref{alg:constant-width-assumption-driver2}.}

\end{algorithm}

\newpage

\appendix
\appendixpage

\section{Proofs of Lemma Statements}\label{appendix:lemma-proofs}

\begin{lemma*}[Restatement of Lemma \ref{thm:mainruntimelemma}]
Let $w$ be defined as in Theorem \ref{thm:inductivestep}. Then $\mathcal{A}_{full}(S=(C,L,M,N), \mathcal{B}, \delta, D)$ runs in time
    $\delta^{-2} \cdot2^{O((d\polylog(n))^{D \cdot 3^{D}}w^{1/3}) (1/\delta)^{1 / \log^2(n)}}$.
\end{lemma*}

\begin{proof}
The runtime analysis of $\mathcal{A}_{full}$ begins the same as in \cite{CC20}. Note that if the IF statement on Line \ref{ln:checkforsmalldelta} is satisfied, then the specified additive error $\delta$ is so small that we can compute the desired quantity, $|\bra{0_{ALL}}C\ket{0_{ALL}}|^2$, exactly, by brute force, in $2^{O(n)}$ time, and this will still take less time than the guaranteed runtime:  
	
	\[T(n) = \delta^{-2} \cdot2^{O((d\polylog(n))^{D \cdot 3^{D}}w^{1/3}) (1/\delta)^{1 / \log^2(n)}}.\]  

Let $T_1(l, D, d, w, \delta)$ represent the run-time of algorithm 1 for a problem with side length $l$ in dimension D with circuit depth $d$ and thickness $w$ in dimensions $> D$ to error $\delta$. Let $T_2(l, D, d, w, \epsilon)$ represent the same for algorithm 2. Then we may bound $T_1$ as follows:

\[T_1(l, D, d, w, \delta) < \frac{n^{1/D}}{10d} T_1(l, D - 1, d, O(w d), E_1(\delta)) + T_2(l, D, d, w, E_2(\delta)),\]
\[T_1(l, 2, d, w, \delta) < \mathcal{B}(n, d, w, \delta)\]
where $E_1(\delta) = 2^{\frac{\log(\delta)}{2h(n)} - 1} - 2^{\frac{\log(\delta)}{h(n)} - 1}$ and $E_2(\delta) = \delta 2^{-10 \log(n) \log(\log(n))}$.
\\
The term $\frac{n^{1/D}}{10d} T_1(l, D - 1, d, O(w d), E_1(\delta))$ follows from lines 6-10 of algorithm 1. This entails making $\frac{n^{1/D}}{10d}$ calls of algorithm 1 on a depth d synthesis in $D - 1$ dimensions to error $E_1(\delta)$ with thickness $O(d)$ in dimension $D$. See the analysis of Theorem 28 of \cite{CC20} for details on how this sub-problem is constructed. The term $T_2(l, D, d, w, E_2(\delta))$ refers to the call of algorithm 2 made in line 14 of algorithm 1. The base case follows directly from line 5 of algorithm 1. By standard recursion analysis, we get that
\\
\[ T_1(l, D, d, w, \delta) < n^{D} \mathcal{B}(n, d, O(wd^D), E_1^{(D - 2)}(\delta)) + \sum_{i = 0}^{D - 3} n^{\frac{i}{D-i+1}} T_2(l, D - i, d, O(wd^i), E_2(E_1^{(i)}(\delta))) \] where $E_1^{(i)}$ refers to the function $E_1$ composed with itself $i$ times. 

Similarly, we can bound $T_2$ as follows:

\[T_2(l, D, d, w, \epsilon) < 2\Delta T_2(\frac{3}{4}l, D, d, w, \epsilon) + \Delta^2 T_1(l, D-1, d^3 \polylog(n), O(wd), \epsilon)\]
\[ +  \Delta^2 2^{\Delta} T_1(l, D-1, d^3 \polylog(n), O(wd), E_3(\epsilon)) + 2 \Delta T_1(l, D - 1, d^2 \polylog(n), O(wd), \epsilon) + \poly(n)\]

\[ T_2(O(1), D, d, w, \epsilon) = T_1(n^{1/D}, D - 1, d, O(w), \epsilon)\]
where $E_3(\epsilon) = \frac{\epsilon}{2^{\Delta}}$.

The term $2\Delta T_2(\frac{3}{4}l, D, d, w, \epsilon)$ follows from the calls to $\mathcal{A}(S_{L,i},\eta-1)$ and $\mathcal{A}(S_{i,R},\eta-1)$ for each $i$. The term $\Delta^2 T_1(l, D-1, d^3 \polylog(n), O(wd), \epsilon)$ refers to the calls to $\mathcal{A}_{full}(S_{i,j}, \mathcal{B},\epsilon, D-1)$ for each $i$ and $j > i$. The term $\Delta^2 2^{\Delta} T_1(l, D-1, d^3 \polylog(n), O(wd), E_3(\epsilon))$ refers to the calls to \\ $\mathcal{A}_{full}\left ( \left (\otimes_{k \in \sigma} \Pi^K_{F_k} \bra{0_{M_k}} \right )\phi_{i,j}\left (\otimes_{k \in \sigma} \ket{0_{M_k}}\Pi^K_{F_k}\right ),\mathcal{B}, \frac{\epsilon}{2^\Delta}, D-1 \right)$ for each $i$, $j > i$, and $\sigma$. The term $2 \Delta T_1(l, D - 1, d^2 \polylog(n), O(wd), \epsilon)$ refers to the calculation of $\kappa_{T, \epsilon}$ for each $i$. For details regarding the construction of the sub-problems for the last three terms, refer to the run-time analysis of algorithm 2 of \cite{CC20}. The final $\poly(n)$ term follows from the calculation of the region $Z$ detailed in line 8 of algorithm 2. The base case follows from the fact that if we have a problem in $D$ dimensions with an $O(1)$ sized edge, we may apply an algorithm in $D - 1$ dimensions to solve it at the cost of an extra $O(1)$ sized thickness. By standard recursion analysis, we get that 

\[T_2(l, D, d, w, \epsilon) < (2\Delta)^\eta \cdot T_2((\frac{3}{4})^\eta l, D, d, w, \epsilon) + \sum^{\eta-1}_{i=0}(2\Delta)^i(\Delta^2 T_1((\frac{3}{4})^il, D-1, d^3 \polylog(n), O(wd), \epsilon) \]
\[+ \Delta^2 2^{\Delta} T_1((\frac{3}{4})^il, D-1, d^3 \polylog(n), O(wd), E_3(\epsilon)) + 2\Delta T_1((\frac{3}{4})^il, D-1, d^2 \polylog(n), O(wd), \epsilon) + \poly(n))\]\\

Now, as we begin to substitute the recurrence relation for $T_2$ (in terms of $T_1)$ into the recurrence relation for $T_1$ (in terms of $T_2$), we need to define $\eta_i$, the number of recursive calls made by $T_2$ to $T_1$ in the $i$-th dimension. Let us define $\eta_i$ as the following:
\begin{equation}\label{eq:eta_dimension}
\begin{split}
    \eta_i = \log_{3/4}(n^{-1/i}) = \frac{\log(n)}{i \cdot \log(4/3)}
\end{split}
\end{equation}
Now that we have defined $\eta_i$, let us substitute the $T_2$ recurrence relation into $T_1$:
\begin{align*}
    T_1(l, D, d, w, \delta) &< n^{D} \mathcal{B}(n, d, O(wd^D), E_1^{(D - 2)}(\delta))\\
    &+ \sum_{i = 0}^{D - 3} n^{\frac{i}{D-i+1}}  \Bigg[(2\Delta)^{\eta_{D-i}} T_1 \Bigg( l, D-i-1, d, O(wd^i), E_2(E_1^{(i)}(\delta)) \Bigg) \\
    &+ \sum^{\eta_{D-i}-1}_{j=0} ( 2\Delta)^j\Bigg(\Delta^2 T_1 \Bigg ( \left(\frac{3}{4}\right)^j l, D-i-1, d^3 \polylog(n), O(wd^{i+1}), E_2(E_1^{(i)}(\delta)) \Bigg) \\
    &+ \Delta^22^\Delta T_1\Bigg(\left(\frac{3}{4}\right)^j l, D-i-1, d^3 \polylog(n), O(wd^{i+1}), E_3(E_2(E_1^{(i)}(\delta)))\Bigg) \\
    &+ 2\Delta T_1\Bigg(\left(\frac{3}{4}\right)^j l, D-i-1, d^2 \polylog(n), O(wd^{i+1}), E_2(E_1^{(i)}(\delta)))\Bigg) + \poly(n)\Bigg)\Bigg] 
\end{align*}
where the first $T_1$ term on the right-hand side comes from unrolling the $T_2$ term in $T_2$'s recurrence relation down to its base-case. 
\\
\\
We can then continue to simplify the upper bound by combining the three terms in the second summation term into $3\Delta^22^\Delta T_1\left(\left(\frac{3}{4}\right)^j l, D-i-1, d^3 \polylog(n), O(wd^{i+1}), E_3(E_2(E_1^{(i)}(\delta)))\right)$ since\\ $3\Delta^22^\Delta \geq (2\Delta + \Delta^22^\Delta + \Delta^2)$ and $E_3(E_2(E_1^{(i)}(\delta))) \leq E_2(E_1^{(i)}(\delta))$.

\begin{align*}
    T_1(l, D, d, w, \delta) &< n^{D} \mathcal{B}(n, d, O(wd^D), E_1^{(D - 2)}(\delta))\\
    &+ \sum_{i = 0}^{D - 3} n^{\frac{i}{D-i+1}}  \Bigg[(2\Delta)^{\eta_{D-i}} T_1 \left( l, D-i-1, d, O(wd^i), E_2(E_1^{(i)}(\delta)) \right) \\
    &+ \sum^{\eta_{D-i}-1}_{j=0} ( 2\Delta)^j\left( 3\Delta^22^\Delta T_1\left(\left(\frac{3}{4}\right)^j l, D-i-1, d^3 \polylog(n), O(wd^{i+1}), E_3(E_2(E_1^{(i)}(\delta)))\right)\right)\Bigg] 
\end{align*}

Next, we can unpack the bracket in the first summation term to get the following:
\begin{align*}
    T_1(l, D, d, w, \delta) &< n^{D} \mathcal{B}(n, d, O(wd^D), E_1^{(D - 2)}(\delta))\\
    &+ \sum_{i = 0}^{D - 3} n^{\frac{i}{D-i+1}} (2\Delta)^{\eta_{D-i}} T_1 \left( l, D-i-1, d, O(wd^i), E_2(E_1^{(i)}(\delta)) \right) \\
    &+ \sum_{i = 0}^{D - 3} \sum^{\eta_{D-i}-1}_{j=0} ( 2\Delta)^{j} n^{\frac{i}{D-i+1}} \Bigg( 3\Delta^2 2^\Delta T_1\Bigg(\Bigg(\frac{3}{4}\Bigg)^j l, D-i-1, d^3 \polylog(n), \\ 
    & O(wd^{i+1}), E_3(E_2(E_1^{(i)}(\delta)))\Bigg)\Bigg)
\end{align*}

The following expression can be obtained by extracting the $T_1$ terms from the summation terms. We do that by bounding all the $T_1$ terms in the first summation term by \\$T_1\left(l, D-1, d, O(wd^D), E_2(E_1^{(D)}(\delta)\right)$ since the runtime will be longer when we start on higher dimension $D$ instead of dimension $D-i-1$, larger thickness $O(wd^D)$ instead of thickness $O(wd^i)$, and smaller error $E_2(E_1^{(D)}(\delta))$ instead of $E_2(E_1^{(i)}(\delta))$. A similar argument could be made for the $T_1$ terms in the second summation term. 

\begin{align*}
    T_1(l, D, d, w, \delta) &< n^{D} \mathcal{B}(n, d, O(wd^D), E_1^{(D - 2)}(\delta))\\
    &+ T_1 \left( l, D-1, d, O(wd^D), E_2(E_1^{(D)}(\delta)) \right) \sum_{i = 0}^{D - 3} n^{\frac{i}{D-i+1}} (2\Delta)^{\eta_{D-i}} \\
    &+ 3\Delta^2 2^\Delta T_1\left(l, D-1, d^3 \polylog(n), O(wd^{D}), E_3(E_2(E_1^{(D)}(\delta)))\right) \sum_{i = 0}^{D - 3} \sum^{\eta_{D-i}-1}_{j=0} ( 2\Delta)^{j} n^{\frac{i}{D-i+1}}
\end{align*}

In the following step, for the first summation term, we bound the $n^{\frac{i}{D-i+1}}$ term by $\poly(n)$ and the $(2\Delta)^{\eta_{D-i}}$ term by $2^{\polylog(n)}$ since $\Delta, \eta = O(\log(n))$. Since we have $O(D)$ terms in the first summation term, we get $O(D \poly(n) 2^{\polylog(n)})$. Likewise, we can do the same thing for the second summation term to get the same upper bound.

\begin{align*}
    T_1(l, D, d, w, \delta) &< n^{D} \mathcal{B}(n, d, O(wd^D), E_1^{(D - 2)}(\delta))\\
    &+ T_1 \left( l, D-1, d, O(wd^D), E_2(E_1^{(D)}(\delta)) \right) O(D \poly(n) 2^{\polylog(n)}) \\
    &+ 3\Delta^2 2^\Delta T_1\left(l, D-1, d^3 \polylog(n), O(wd^{D}), E_3(E_2(E_1^{(D)}(\delta)))\right) O(D \poly(n) 2^{\polylog(n)})
\end{align*}

The following expression can be obtained by combining the second and third term in the previous expression. We get $T_1\left(l, D-1, d^3\polylog(n), O(wd^D), E_3(E_2(E_1^{(D)}(\delta)))\right)$ since\\ $d^3\polylog(n) > d$ and $E_3(E_2(E_1^{(i)}(\delta))) \leq  E_2(E_1^{(i)}(\delta))$ which would give us a larger runtime bound. The $3\Delta^2 2^\Delta $ can be absorbed into the $O(D\poly(n)2^{\polylog(n)})$ term.

\begin{align*}
    T_1(l, D, d, w, \delta) &< n^{D} \mathcal{B}(n, d, O(wd^D), E_1^{(D - 2)}(\delta))\\
    &+ T_1\left(l, D-1, d^3 \polylog(n), O(wd^{D}), E_3(E_2(E_1^{(D)}(\delta)))\right) O(D \poly(n) 2^{\polylog(n)})
\end{align*}

Now, we substitute the BGM algorithm's runtime from Theorem 5 of \cite{BGM19} into the first term to get the recurrence for $T_1$ in dimension $D$ in terms of $T_1$ in one dimension lower.

\begin{align*}
    T_1(l, D, d, w, \delta) &< \poly(n^D) (E_1^{(D - 2)}(\delta))^{-2} 2^{d^3 w^{1/D}}\\
    &+ T_1\left(l, D-1, d^3 \polylog(n), O(wd^{D}), E_3(E_2(E_1^{(D)}(\delta)))\right) O(D \poly(n) 2^{\polylog(n)})
\end{align*}

Before we begin to unroll the recurrence relation for $T_1$ with respect to its dimension, let us first define $f(d)$, the depth of the block-encoding (at this point in the analysis), and $f^{(k)}(d)$, the depth of the block-encoding after unrolling the recurrence relation for $k$ dimensions

\begin{align*}
   f(d) &= d^3 \polylog(n) \\
   f^{(k)}(d) &< d^{3^k} (\polylog(n))^{3^k} 
\end{align*}

We can also define $g(d,w)$, the thickness of the circuit (at this point in the analysis), and $g^{(k)}(d)$, the thickness of the circuit after unrolling the recurrence relation for $k$ dimensions as follows:

\begin{align*}
  g(d, w) &= O(wd^D) \\
  g^{(k)}(d, w) &= g^{(k - 1)}(d, w) (f^{(k-1)}(d))^D \\
 &= g^{(k - 2)}(d, w) (f^{(k-2)}(d))^D (f^{(k-1)}(d))^D\\
&= g^{(k - \ell)}(d, w) \prod_{i = 1}^{\ell} (f^{(k - i)}(d))^D \\
&= g(d, w) \prod_{i = 1}^{k - 1} (f^{(k - i)}(d))^D\\
&< O(wd^D) (\prod_{i = 1}^{k - 1} (d \polylog(n))^{3^i})^D\\
&< O(wd^D) (d \polylog(n))^{D3^k}
\end{align*}

Now, we want to write the unrolling of the recurrence relation for $T_1$ with respect to dimensions in terms of $f(d)$ and $g(d,w)$. To simplify the writing, we define $E_5(\delta) = E_3(E_2(E_1^{(D)}(\delta)))$. The following expression is obtained by unrolling $T_1$'s recurrence relation for an arbitrary dimension $D$ to dimension 2 which is the base-case for $T_1$.

\begin{align*}
    T_1(\ell, D, d, w, \delta) &<
    O((D \poly(n^D) 2^{\polylog(n)})^{D - 2}) T_1(\ell, 2, f^{(D - 2)}(d), g^{(D-2)}(d,w), E_5^{(D-2)}(\delta)) \\
    &+ \sum^{D-3}_{i=0} O((D \poly(n^D) 2^{\polylog(n)})^{i}) \poly(n) \left(E_1^{(D-i-2)}\left(E_5^{(i)}\left(\delta\right)\right)\right)^{-2} 2^{(f^{(i)}(d))^3(g^{(i)}(d, w))^{\frac{1}{D - i}}}
\end{align*}

To further simplify, we replace each occurrence of $i$ in order to maximize each quantity, then replace each occurrence of $D - 2$ with $D$. Note that the more we compose $E_1$ and $E_5$, the smaller they get and hence their inverse-squared form will be larger. For $f(d)$ and $g(d,w)$, the more we composed them, the greater the depth and the thicker the thickness of the block-encoding which gives us an upper bound for the runtime. Note how we chose the upper bound for the exponent of the $g^{(D)}(d,w)$ to be $\frac{1}{3}$ to get the smallest root form to maximize the exponent of the $2$ term.  

\begin{align*}
    T_1(\ell, D, d, w, \delta) &<
    O((D \poly(n) 2^{\polylog(n^D)})^{D}) T_1(\ell, 2, f^{(D)}(d), g^{(D)}(d,w), E_5^{(D)}(\delta)) \\
    &+ D\cdot O((D \poly(n^D) 2^{\polylog(n)})^{D}) \poly(n) \left(E_1^{(D)}\left(E_5^{(D)}\left(\delta\right)\right)\right)^{-2} 2^{(f^{(D)}(d))^3(g^{(D)}(d, w))^{\frac{1}{3}}}
\end{align*}

Next, we write the first term of the right-hand side of the first inequality according to BGM's runtime as given in Theorem 5 of \cite{BGM19} and then brought the $D\cdot \poly(n)$ coefficient in the second term into the second term's big-$O$. The second inequality comes from the fact that the first term is smaller than the second term and hence can be absorbed into the second term.

\begin{align*}
    T_1(\ell, D, d, w, \delta) &<
    O((\poly(n^D))^{D + 1}(D 2^{\polylog(n)})^{D}) (E_5^{(D)}(\delta))^{-2} 2^{(f^{(D)}(d))^2 (g^{(D)}(d,w))^{1/D} } \\
    &+ O((D \poly(n^D))^{D + 1}(2^{\polylog(n)})^{D}) \left(E_1^{(D)}\left(E_5^{(D)}\left(\delta\right)\right)\right)^{-2} 2^{(f^{(D)}(d))^3(g^{(D)}(d, w))^{\frac{1}{3}}}\\
     &<O((D \poly(n^D))^{D + 1}(2^{\polylog(n)})^{D}) \left(E_1^{(D)}\left(E_5^{(D)}\left(\delta\right)\right)\right)^{-2} 2^{(f^{(D)}(d))^3(g^{(D)}(d, w))^{\frac{1}{3}}}
\end{align*}

Now, we substitute the upper bounds for $f^{(k)}(d)$ and $g^{(k)}(d, w)$ as previously defined into the above expression to get the following inequality:

\begin{align*}
    T_1(\ell, D, d, w, \delta) &<
     O((D \poly(n^D))^{D + 1}(2^{\polylog(n)})^{D}) \left(E_1^{(D)}\left(E_5^{(D)}\left(\delta\right)\right)\right)^{-2} \\
     &\cdot2^{(d\polylog(n))^{3^{D + 1}}(O(wd^D)^\frac{1}{3} (d \polylog(n))^{D3^{D-1}})}\\
     &=
     O((D \poly(n^D))^{D + 1}(2^{\polylog(n)})^{D}) \left(E_1^{(D)}\left(E_5^{(D)}\left(\delta\right)\right)\right)^{-2} \\
     &\cdot2^{O(wd^D)^{\frac{1}{3}} (d \polylog(n))^{(9+D)3^{D-1}}}\\
     &<     O((D \poly(n^D))^{D + 1}(2^{\polylog(n)})^{D}) \left(E_1^{(D)}\left(E_5^{(D)}\left(\delta\right)\right)\right)^{-2} \\
     &\cdot2^{O(d^{\left(3^{D+1} + D / 3 + D3^{D - 1}\right)} (\polylog(n))^{\left(3^{D + 1} + D3^{D - 1}\right)} w^{1/3})} \\
     &< \left(E_1^{(D)}\left(E_5^{(D)}\left(\delta\right)\right)\right)^{-2} \cdot2^{O(d^{\left(3^{D+1} + D / 3 + D3^{D - 1}\right)} (\polylog(n))^{\left(3^{D + 1} + D3^{D - 1}\right)} w^{1/3})} \\
     &< \left(E_1^{(D)}\left(E_5^{(D)}\left(\delta\right)\right)\right)^{-2} \cdot2^{O((d\polylog(n))^{D3^{D}} w^{1/3})}
\end{align*}

Now note that \\

$E_1(\delta) = 2^{\frac{\log(\delta)}{2h(n)} - 1} - 2^{\frac{\log(\delta)}{h(n)} - 1} = \frac{1}{2} (2^{\frac{\log(\delta)}{2h(n)} } - 2^{\frac{\log(\delta)}{h(n)}}) \geq \frac{\ln{2}}{2} 2^{\frac{\log(\delta)}{h(n)}} (\frac{\log(\delta)}{2h(n)}  -\frac{\log(\delta)}{h(n)} ) = -\frac{\log(\delta)}{4h(n)} 2^{\frac{\log(\delta)}{h(n)}} \cdot \ln{2} $\\

Hence, by monotonicity, 
\begin{align*}
    E_1(E_1(\delta))  = -\frac{\log(E_1(\delta))}{4h(n)} 2^{\frac{\log(E_1(\delta))}{h(n)}} \cdot \ln{2}    
    &\geq -\frac{\log(-\frac{\log(\delta)}{4h(n)}2^{\frac{\log(\delta)}{h(n)}} \cdot \ln{2})}{4h(n)} 2^{\frac{\log(-\frac{\log(\delta)}{4h(n)} 2^{\frac{\log(\delta)}{h(n)}} \cdot \ln{2})}{h(n)}} \cdot \ln{2}\\
    &= -\frac{\log(-\frac{\log(\delta)}{4h(n)}) + \log(2^{\frac{\log(\delta)}{h(n)}} ) + \log( \ln{2}))}{4h(n)} 2^{\frac{\log(-\frac{\log(\delta)}{4h(n)} 2^{\frac{\log(\delta)}{h(n)}} \cdot \ln{2})}{h(n)}} \cdot \ln{2}  \\
    &\geq -\frac{ \log( \ln{2})}{4h(n)} 2^{\frac{\log(-\frac{\log(\delta)}{4h(n)} 2^{\frac{\log(\delta)}{h(n)}} \cdot \ln{2})}{h(n)}} \cdot \ln{2}  \\
    &= -\frac{ \log( \ln{2})}{4h(n)} 2^{\frac{\log(E_1(\delta))}{h(n)}} \cdot \ln{2}  \\
    &\geq -\frac{ \log( \ln{2})}{4h(n)} (E_1(\delta))^{\frac{1}{h(n)}} \cdot \ln{2}\\
    &\geq -\frac{ \log( \ln{2})}{4h(n)} E_1(\delta) \cdot \ln{2}
\end{align*}

And so for some constant $a$ we get, 

\begin{align*}
    E_1(a \delta)  \geq -\frac{ \log( \ln{2})}{4h(n)} a \delta \cdot \ln{2}
\end{align*}

Therefore 
\begin{align*}
    &E_1^{(D)}(\delta) \geq (\frac{-\log(\ln(2))\ln(2)}{4 h(n)})^{D-1} E_1(\delta) \geq (\frac{-\log(\ln(2))\ln(2)}{4 h(n)})^{D} \delta \\
    &\implies E_5(\delta) = E_3(E_2(E_1^{(D)}(\delta))) \geq 2^{-10\log(n)\log(\log(n)) - \Delta} (\frac{-\log(\ln(2))\ln(2)}{4 h(n)})^{D} \delta \\
    &\implies E_5^{(D)}(\delta) \geq 2^{D(-10\log(n)\log(\log(n)) - \Delta)} \left(\frac{-\log(\ln(2))\ln(2)}{4 h(n)} \right)^{D^2} \delta \\
    &\implies (E_1^{(D)} \circ E_5^{(D)}) (\delta) 
    \geq 2^{D(-10\log(n)\log(\log(n)) - \Delta)} \left(\frac{-\log(\ln(2))\ln(2)}{4 h(n)} \right)^{D^2 + D} \delta
\end{align*}

Thus with $\Delta = \log(n)$ we get that 

\begin{align*}
\left( (E_1^{(D)} \circ E_5^{(D)}) (\delta)\right)^{-2} &\leq 2^{D(10\log(n)\log(\log(n)) + \Delta)} \left(\frac{-\log(\ln(2))\ln(2)}{4 h(n)} \right)^{-2D^2 - 2D} \delta^{-2} \\
&= 2^{D \polylog(n)} \left(\frac{-\log(\ln(2))\ln(2)}{4 h(n)} \right)^{-2D^2 - 2D} \delta^{-2}
\end{align*}

Plugging this into our run-time bound we get

\begin{align*}
    T_1(\ell, D, d, w, \delta) &< 2^{D \polylog(n)} \left(\frac{-\log(\ln(2))\ln(2)}{4 h(n)} \right)^{-2D^2 - 2D} \delta^{-2} \cdot2^{O((d\polylog(n))^{D3^{D}} w^{1/3})} \\
    &= \delta^{-2} \cdot2^{O((d\polylog(n))^{D3^{D}} w^{1/3})}
\end{align*}
\end{proof}

\begin{lemma*} [Restatement of Lemma \ref{thm:mainerrorlemma}]
    $\mathcal{A}_{full}(S=(C,L,M,N), \mathcal{B}, \delta, D)$ returns an $\delta$-additive error approximation of $|\bra{0_{ALL}}C\ket{0_{ALL}}|^2$ 
\end{lemma*}
\begin{proof}
 The error analysis of the error obtained by $\mathcal{A}_{full}(S, \mathcal{B}, \delta, D)$ can be broken into four cases according to the IF statements on Lines \ref{ln:checkforsmalldelta}, \ref{ln:deltalarge}, \ref{ln:easyifalg1}, and \ref{ln:hardifalg1} of Algorithm \ref{alg:quasi-poly-driver}. The first three cases can be easily shown to return the value in $\delta$-additive error within the promised runtime as shown in page 16 and 20 of \cite{CC20}. 
	
In the event that Line \ref{ln:hardifalg1} is satisfied, Algorithm \ref{alg:quasi-poly-driver} returns the following quantity:
\[\mathcal{A}(S, \eta = \frac{\log(n)}{D\log(4/3)} , \Delta = \log(n), \epsilon = \delta2^{-10 \log(n) \log(\log(n)))} , h(n) = \log^7(n), K_{heavy} , D,  \mathcal{B}),\]
which we know is an $f(S,\eta_D,\Delta,\epsilon, D)$-additive error approximation of $|\bra{0_{ALL}}C\ket{0_{ALL}}|^2$. Recall the definition of $\eta_D$ defined in Equation \ref{eq:eta_dimension}. Since $\eta_D =\frac{\log(n)}{D\log(4/3)}$, by Equation \ref{eq:error-recursive-bound}, we know that:

	\begin{align*}
	f(S,\eta,\Delta,\epsilon) &\leq \eta_D(20\Delta^2)^{\eta_D}\left((2e(n) + 2g(n))^\Delta + 3\Delta^2\mathcal{E}_3(n, K, T, \epsilon, \Delta)\right) \\
	&= \eta_D(20\Delta^2)^{\eta_D}3\Delta^2 O\left(\mathcal{E}_3(n, K, T, \epsilon, \Delta)\right) \\
	&= \eta_D(20\Delta^2)^{\eta_D}3\Delta^2O\left(2^\Delta(2e(n))^K + 2^\Delta K\left(e(n)^{2T} + \epsilon\right) + \epsilon\right)\\
		&= \frac{\log(n)}{D\log(4/3)}\left(20\log^2(n)\right)^{\frac{\log(n)}{D\log(4/3)}}3\log^2(n) O\left(2^{\log(n)} (2(1 - 2^{\frac{\log(\delta)}{\log^7(n)}}))^{\log^3(n)} \right .\\		& \left .+ 2^{\log(n)} \log^3(n)\left ((1 - 2^{\frac{\log(\delta)}{\log^7(n)}})^{2\log^3(n)}+\epsilon \right ) +  \delta 2^{-10 \log(n) \log(\log(n))} \right )\\
		& \leq  (\log(n))^{2\log(n)} \cdot \poly(n) \cdot 	\left ( (2(1 - 2^{\frac{\log(\delta)}{\log^7(n)}}))^{\log^3(n)} +\epsilon  + \delta 2^{-10 \log(n) \log(\log(n))} \right )\\
		&  \leq  (\log(n))^{2\log(n)} \cdot \poly(n) \cdot \left (	\left (  O\left (\frac{1}{\log^4(n)}\right )\right )^{\log^3(n)} +2 \cdot \delta 2^{-10 \log(n) \log(\log(n))} \right )\\
		&\leq  2^{2\log(n)\log(\log(n))} \cdot \poly(n) \cdot 	\left (  O\left (\frac{1}{\log^4(n)}\right )\right )^{\log^3(n)} + \delta 2^{-8 \log(n) \log(\log(n))} \\
	&  \leq  o(1) \cdot \delta + o(1)  \cdot \delta =  o(1)  \cdot \delta \numberthis \label{eq:upper-bound-even-spaced-larger2}
	\end{align*}
where the first inequality follows from our result from the next subsection and the rest follows by calculation, noting that $E_3(n, K, T, \epsilon,\Delta) \geq  (2e(n)+2\projerror)^{\Delta}$ for our specific choice of parameters (in particular $\Delta = \log(n)$). Note from \cite{CC20} that $e(n) \leq (1 - 2^{\frac{\log(\delta)}{\log^7(n)}}) = O(1/\log^4(n))$   (since $\delta \geq n^{-\log^2(n)} = 2^{-\log(n)^3}$ as verified in Algorithm \ref{alg:quasi-poly-driver}), $K = \log^3(n)$, $T = \log^3(n)$, and $\epsilon = \delta 2^{-10 \log(n) \log(\log(n))}$.  The final inequality, which claims $ 2^{2\log(n)\log(\log(n))} \cdot \poly(n) \cdot 	\left (  O\left (\frac{1}{\log^4(n)}\right )\right )^{\log^4(n)} = o(1) \cdot \delta$, again follows because $\delta \geq n^{-\log^2(n)}$ as verified in the driver algorithm, Algorithm \ref{alg:quasi-poly-driver}.

As described on page 22 of $\cite{CC20}$, $\bra{0_{ALL}}\ket{\Psi_\emptyset}\bra{\Psi_\emptyset}\ket{0_{ALL}}$  is the quantity that we wish for Algorithm \ref{alg:quasi-poly-subroutine} to output. Refer to Definition 17 and Lemma 18 from \cite{CC20} for the definition of $\ket{\Psi_{\emptyset}}$ and $\ket{\Psi_{\sigma}}$ for the subsequent analysis. Since Algorithm \ref{alg:quasi-poly-subroutine} depends on recursively calling itself, recall $\eta_{i}$ from Equation \ref{eq:eta_dimension} that defines the number of recursive calls for some dimension $i$. The error between the returned output of Algorithm \ref{alg:quasi-poly-subroutine}, (defined on Line \ref{ln:returnofA} of that algorithm) and the desired output quantity $\bra{0_{ALL}}\ket{\Psi_\emptyset}\bra{\Psi_\emptyset}\ket{0_{ALL}}$ is written below:

\begin{align}
{}&f(S,\eta_D,\Delta,\epsilon, D) \leq \Big\| \bra{0_{ALL}}\ket{\Psi_\emptyset}\bra{\Psi_\emptyset}\ket{0_{ALL}} - \mathcal{A}(S, \eta_{D}, D, \epsilon) \Big\| \\
{}&\leq \Big\| \bra{0_{ALL}}\ket{\Psi_\emptyset}\bra{\Psi_\emptyset}\ket{0_{ALL}} - \sum_{\sigma\in\mathcal{P}([\Delta])\setminus\emptyset} (-1)^{\abs{\sigma}+1} \bra{0_{ALL}}\ket{\Psi_\sigma}\bra{\Psi_\sigma}\ket{0_{ALL}} \Big\| \\
&+ \Big\| \sum_{\sigma\in\mathcal{P}([\Delta])\setminus\emptyset} (-1)^{\abs{\sigma}+1} \bra{0_{ALL}}\ket{\Psi_\sigma}\bra{\Psi_\sigma}\ket{0_{ALL}} - \mathcal{A}(S, \eta_{D}, D, \epsilon) \Big\| \\
{}&\leq \left (2e(n)+2\projerror \right)^{\Delta} + \Big\| \sum_{\sigma\in\mathcal{P}([\Delta])\setminus\emptyset} (-1)^{\abs{\sigma}+1} \bra{0_{ALL}}\ket{\Psi_\sigma}\bra{\Psi_\sigma}\ket{0_{ALL}} - \mathcal{A}(S, \eta_{D}, D, \epsilon) \Big\| \\
& = \left (2e(n)+2\projerror \right)^{\Delta} + \Bigg\|\sum_{\sigma\in\mathcal{P}([\Delta])\setminus\emptyset} (-1)^{\abs{\sigma}+1} \bra{0_{ALL}}\ket{\Psi_\sigma}\bra{\Psi_\sigma}\ket{0_{ALL}} \\
&-\Bigg( \sum_{i=1}^{\Delta}\frac{1}{(\kappa^{i}_{T, \epsilon})^{4K+1}}\mathcal{A}(S_{L,i},\eta_D-1, D, \epsilon) \cdot \mathcal{A}(S_{i,R},\eta_D-1, D, \epsilon) \\
&-\sum_{i=1}^{\Delta}\sum_{j=i+1}^{\Delta}\frac{1}{(\kappa^{i}_{T, \epsilon}\kappa^{j}_{T, \epsilon})^{4K+1}}\mathcal{A}(S_{L,i},\eta_D-1, D, \epsilon) \cdot \mathcal{A}_{full}(S_{i,j},\eta_{D-1}, D - 1, \epsilon)\cdot \mathcal{A}(S_{j,R},\eta_D-1, D, \epsilon) \\
&+\sum_{i=1}^{\Delta}\sum_{j=i+2}^{\Delta}\frac{1}{(\kappa^{i}_{T, \epsilon}\kappa^{j}_{T, \epsilon})^{4K+1}}\mathcal{A}(S_{L,i},\eta_D-1, D, \epsilon) \cdot \mathcal{A}(S_{j,R},\eta_D-1, D, \epsilon) \\
&\cdot\Bigg[\sum_{\sigma\in\mathcal{P}(\{i+1,\dots,j-1\})\setminus\emptyset}(-1)^{\abs{\sigma}+1} \mathcal{A}_{full}\left ( \left (\otimes_{k \in \sigma} \Pi^K_{F_k} \bra{0_{M_k}} \right )\phi_{i,j}\left (\otimes_{k \in \sigma} \ket{0_{M_k}}\Pi^K_{F_k}\right ), \eta_{D-1}, D - 1, E_3(\epsilon) \right) \Bigg]
\Bigg)\Bigg\|
\end{align}

Grouping analogous terms and using triangle inequality gives:

\begin{align*}
&f(S,\eta_D,\Delta,\epsilon, D) \leq \left (2e(n)+2\projerror \right)^{\Delta} \\
& +  \Bigg \|  \sum_{i=1}^{\Delta}\left ( \bra{0_{ALL}}\ket{\Psi_{\{i\}}}\bra{\Psi_{\{i\}}}\ket{0_{ALL}} -  \frac{1}{(\kappa^{i}_{T, \epsilon})^{4K+1}}\mathcal{A}(S_{L,i},\eta_D-1, D, \epsilon) \cdot \mathcal{A}(S_{i,R},\eta_D-1, D, \epsilon)\right )\\
&+\sum_{i=1}^{\Delta}\sum_{j=i+1}^{\Delta} \Bigg(\frac{1}{(\kappa^{i}_{T, \epsilon}\kappa^{j}_{T, \epsilon})^{4K+1}}\mathcal{A}(S_{L,i},\eta_D-1, D, \epsilon) \cdot \mathcal{A}_{full}(S_{i,j},\mathcal{B}, D - 1, \epsilon)\cdot \mathcal{A}(S_{j,R},\eta_D-1, D, \epsilon) \\
&- \bra{0_{ALL}}\ket{\Psi_{\{i,j\}}}\bra{\Psi_{\{i,j\}}}\ket{0_{ALL}} \Bigg) -\sum_{i=1}^{\Delta}\sum_{j=i+2}^{\Delta}\sum_{\sigma\in\mathcal{P}(\{i+1,\dots,j-1\})\setminus\emptyset}\Bigg ( \frac{1}{(\kappa^{i}_{T, \epsilon}\kappa^{j}_{T, \epsilon})^{4K+1}}\mathcal{A}(S_{L,i},\eta_D-1, D, \epsilon) \\
&\cdot\,\mathcal{A}(S_{j,R},\eta_D-1, D, \epsilon) \cdot(-1)^{\abs{\sigma}+1} \mathcal{A}_{full}\left ( \left (\otimes_{k \in \sigma} \Pi^K_{F_k} \bra{0_{M_k}} \right )\phi_{i,j}\left (\otimes_{k \in \sigma} \ket{0_{M_k}}\Pi^K_{F_k}\right ), \mathcal{B}, D - 1, E_3(\epsilon) \right)  \\
&- (-1)^{|\sigma| + 1}\bra{0_{ALL}}\ket{\Psi_{\{i,j\}\cup \sigma}}\bra{\Psi_{\{i,j\}\cup \sigma}}\ket{0_{ALL}} \Bigg)
\Bigg\|\\
& \leq \left (2e(n)+2\projerror \right)^{\Delta} \\
& +  \sum_{i=1}^{\Delta}\Bigg \|  \left ( \bra{0_{ALL}}\ket{\Psi_{\{i\}}}\bra{\Psi_{\{i\}}}\ket{0_{ALL}} - \frac{1}{(\kappa^{i}_{T, \epsilon})^{4K+1}}\mathcal{A}(S_{L,i},\eta_D-1, D, \epsilon) \cdot \mathcal{A}(S_{i,R},\eta_D-1, D, \epsilon) \right )  \Bigg \|\\
&+  \sum_{i=1}^{\Delta}\sum_{j=i+1}^{\Delta}\Bigg \| \Bigg (\frac{1}{(\kappa^{i}_{T, \epsilon}\kappa^{j}_{T, \epsilon})^{4K+1}}\mathcal{A}(S_{L,i},\eta_D-1, D, \epsilon) \cdot \mathcal{A}_{full}(S_{i,j},\mathcal{B}, D - 1, E_3(\epsilon))\cdot \mathcal{A}(S_{j,R},\eta_D-1, D, \epsilon)\\
&- \bra{0_{ALL}}\ket{\Psi_{\{i,j\}}}\bra{\Psi_{\{i,j\}}}\ket{0_{ALL}} \Bigg )  \Bigg \|\\
&-\sum_{i=1}^{\Delta}\sum_{j=i+2}^{\Delta}\Bigg \| \sum_{\sigma\in\mathcal{P}(\{i+1,\dots,j-1\})\setminus\emptyset}(-1)^{|\sigma|+1} \Bigg(  \frac{1}{(\kappa^{i}_{T, \epsilon}\kappa^{j}_{T, \epsilon})^{4K+1}}\mathcal{A}(S_{L,i},\eta_D-1, D, \epsilon) \cdot \mathcal{A}(S_{j,R},\eta_D-1, D, \epsilon)\\
&\cdot  \mathcal{A}_{full}\left ( \left (\otimes_{k \in \sigma} \Pi^K_{F_k} \bra{0_{M_k}} \right )\phi_{i,j}\left (\otimes_{k \in \sigma} \ket{0_{M_k}}\Pi^K_{F_k}\right ), \mathcal{B}, D - 1, E_3(\epsilon) \right) - \bra{0_{ALL}}\ket{\Psi_{\{i,j\}\cup \sigma}}\bra{\Psi_{\{i,j\}\cup \sigma}}\ket{0_{ALL}} \Bigg)
\Bigg\| \numberthis \label{eq:threesplits}
\end{align*}

We will now use Lemma \ref{lem:singlecuttrick}, \ref{lem:doublecuttrick}, and \ref{lem:multicuttrick} that are adapted versions of Lemma 29, 30, and 31 from \cite{CC20} to bound the last three terms of the above inequality. Because their bounds are independent of dimensions, the proofs for the three lemmas will be similar to the proofs in \cite{CC20}. 

\begin{equation}\label{eq:error-recursive-bound}
\begin{split}
f(S,\eta_D,\Delta,\epsilon, D) &\leq \left (2e(n)+2\projerror \right)^{\Delta}+ \Delta\left(\mathcal{E}_1(n, K, T, \epsilon) +2f(S,\eta_D-1,\Delta,\epsilon, D)\right) \\
&+ \Delta^2\left(\mathcal{E}_2(n, K, T, \epsilon) +2f(S,\eta_D-1,\Delta,\epsilon, D)\right)\\
&+\Delta^2\left(\mathcal{E}_3(n, K, T, \epsilon, \Delta) +16 f(S,\eta_D-1,\Delta,\epsilon, D) \right) \\
&\leq \left (2e(n)+2\projerror \right)^\Delta + 3\Delta^2\mathcal{E}_3(n, K, T, \epsilon, \Delta) + 20\Delta^2f(S,\eta_D-1,\Delta,\epsilon, D) \\
&= (2e(n) + 2g(n))^\Delta + 3\Delta^2\mathcal{E}_3(n, K, T, \epsilon, \Delta) \\
&+ 20\Delta^2\Bigg[(2e(n) + 2g(n))^\Delta + 3\Delta^2\mathcal{E}_3(n, K, T, \epsilon, \Delta) + 20\Delta^2 f(S, \eta_D - 2, \Delta, \epsilon, D)\Bigg]\\
&=\sum_{i=0}^{\eta_D - 1}\Bigg[(20\Delta^2)^i\Bigg((2e(n) + 2g(n))^\Delta + 3\Delta^2\mathcal{E}_3(n, K, T, \epsilon, \Delta)\Bigg)\Bigg] + (20\Delta^2)^{\eta_D}f(S, 0, \Delta, \epsilon, D)\\
&\leq\eta_{D}(20\Delta^2)^{\eta_{D}}\Bigg((2e(n) + 2g(n))^\Delta + 3\Delta^2\mathcal{E}_3(n, K, T, \epsilon, \Delta)\Bigg) + (20\Delta^2)^{\eta_D}\epsilon \\
&\leq \eta_D(20\Delta^2)^{\eta_D}\left(\epsilon + (2e(n) + 2g(n))^\Delta + 3\Delta^2\mathcal{E}_3(n, K, T, \epsilon, \Delta)\right)\\
&\leq \eta_D(20\Delta^2)^{\eta_D}\left((2e(n) + 2g(n))^\Delta + 3\Delta^2\mathcal{E}_3(n, K, T, \epsilon, \Delta)\right)
\end{split}
\end{equation}

where the above inequalities follow because $\mathcal{E}_3(n, K, T, \epsilon, \Delta) \geq \mathcal{E}_2(n, K, T, \epsilon) \geq \mathcal{E}_1(n, K, T, \epsilon)$ and $f(S, 0, \Delta, \epsilon, \cdot) \leq \epsilon \leq \mathcal{E}_3(n, K, T, \epsilon, \Delta)$

\end{proof}

\begin{lemma} \label{lem:singlecuttrick}
	
	\begin{align*}
&\Bigg \|  \left ( \frac{1}{(\kappa^{i}_{T, \epsilon})^{4K+1}}\mathcal{A}(S_{L,i},\eta_D-1, D, \epsilon) \cdot \mathcal{A}(S_{i,R},\eta_D-1, D, \epsilon)  - \bra{0_{ALL}}\ket{\Psi_{\{i\}}}\bra{\Psi_{\{i\}}}\ket{0_{ALL}} \right )  \Bigg \| \\
&\leq \mathcal{E}_1(n, K, T, \epsilon) +2f(S,\eta_D-1,\Delta,\epsilon, D),
	\end{align*}
	
	where $\mathcal{E}_1(n, K, T, \epsilon) \equiv 10K(e(n)^{2T} + 6\projerror +  \epsilon)$.
\end{lemma}

\begin{lemma}\label{lem:doublecuttrick}
		\begin{equation}
		\begin{split}
		&\Bigg \| \left (\frac{1}{(\kappa^{i}_{T, \epsilon}\kappa^{j}_{T, \epsilon})^{4K+1}}\mathcal{A}(S_{L,i},\eta_D-1, D, \epsilon) \cdot \mathcal{A}_{full}(S_{i,j},\mathcal{B}, D - 1, \epsilon)\cdot \mathcal{A}(S_{j,R},\eta_D-1, D, \epsilon) \right.\\
		&- \left.\bra{0_{ALL}}\ket{\Psi_{\{i,j\}}}\bra{\Psi_{\{i,j\}}}\ket{0_{ALL}} \right )  \Bigg \| \\
		&\leq \mathcal{E}_2(n, K, T,\epsilon) +2f(S,\eta_D-1,\Delta,\epsilon, D),
		\end{split}
		\end{equation}
		
		where $\mathcal{E}_2(n, K, T, \epsilon) \equiv 10K(e(n)^{2T} + 6\projerror +  \epsilon) + \epsilon $
\end{lemma}

\begin{lemma} \label{lem:multicuttrick}
	\begin{align*}
	&\Bigg \| \sum_{\sigma\in\mathcal{P}(\{i+1,\dots,j-1\})\setminus\emptyset}(-1)^{|\sigma|+1} \left (  \frac{1}{(\kappa^{i}_{T, \epsilon}\kappa^{j}_{T, \epsilon})^{4K+1}}\mathcal{A}(S_{L,i},\eta_D-1, D, \epsilon) \cdot \mathcal{A}(S_{j,R},\eta_D-1, D, \epsilon) \right .\\
& \cdot \mathcal{A}_{full}\left ( \left (\otimes_{k \in \sigma} \Pi^K_{F_k} \bra{0_{M_k}} \right )\phi_{i,j}\left (\otimes_{k \in \sigma} \ket{0_{M_k}}\Pi^K_{F_k}\right ), \mathcal{B}, D - 1, E_3(\epsilon) \right)  \\
&\left .- \bra{0_{ALL}}\ket{\Psi_{\{i,j\}\cup \sigma}}\bra{\Psi_{\{i,j\}\cup \sigma}}\ket{0_{ALL}} \right )
\Bigg\| \numberthis\\
	&\leq \mathcal{E}_3(n, K, T, \epsilon, \Delta) + 16f(S,\eta_D-1,\Delta,\epsilon, D),
	\end{align*}
	
	where 
	
	\begin{align*}
	& \mathcal{E}_3(n, K, T, \epsilon, \Delta) \equiv   O \left (2^{\Delta} (6\projerror) + 2^\Delta K \left (e(n)^{2T}+\epsilon \right ) + \epsilon\right) 
	\end{align*}
\end{lemma}
\section*{Acknowledgment}

MC thanks Sergey Bravyi for helpful discussions.  We thank Gorjan Alagic and Nolan Coble for attending and contributing to some project group meetings.

\end{document}